\documentclass[12pt]{article}
\usepackage{latexsym,amssymb,upref,amsmath,amsthm, amsfonts}
\usepackage[bookmarks=true, pdftoolbar=true]{hyperref}

\newtheorem{theorem}{Theorem}[section]
\newtheorem{lemma}[theorem]{Lemma}
\newtheorem{proposition}[theorem]{Proposition}
\newtheorem{corollary}[theorem]{Corollary}
\newtheorem{claim}[theorem]{Claim}

\theoremstyle{definition}

\newtheorem*{example}{Example}

\numberwithin{equation}{section}

\newcommand\mc[1]{\mathcal{#1}}

\newcommand\Bin{\text{Bin}}

\title{The phase transition in inhomogeneous random intersection graphs}

\author{Milan Bradonji\'c\thanks{Mathematics of Networks and Communications, Bell Labs, Alcatel-Lucent, 600 Mountain Avenue, Murray Hill, New Jersey 07974, USA; \texttt{milan@research.bell-labs.com}. Research supported in part by NIST grant 60NANB10D128. Part of this work was done at Los Alamos National Laboratory.}
\and Aric Hagberg\thanks{Center for Nonlinear Studies and Theoretical Division, Los Alamos National Laboratory, Los Alamos, NM 87545, USA;~\texttt{hagberg@lanl.gov}.}
\and Nicolas W. Hengartner\thanks{Information Sciences Group, Los Alamos National Laboratory, Los Alamos, NM 87545, USA;~\texttt{nickh@lanl.gov}.}
\and Nathan Lemons\thanks{Center for Nonlinear Studies and Theoretical Division, Los Alamos National Laboratory, Los Alamos, NM 87545, USA;~\texttt{nlemons@lanl.gov}.}
\and Allon G. Percus\thanks{School of Mathematical Sciences, Claremont Graduate University, Claremont, CA 91711, USA;~\texttt{allon.percus@cgu.edu}.}
}

\begin{document}
\maketitle

\begin{abstract}
We analyze the component evolution in inhomogeneous random intersection graphs when the average degree is close to 1.  As the average degree increases, the size of the largest component in the random intersection graph goes through a phase transition.  We give bounds on the size of the largest components before and after this transition.   We also prove that the largest component after the transition is unique.  These results are similar to the phase transition in Erd\H os-R\'enyi random graphs; one notable difference is that the jump in the size of the largest component varies in size depending on the parameters of the random intersection graph. 

\textbf{Keywords:} Random intersection graphs, random graphs, giant component, phase transition, branching process.
\end{abstract}

\section{Introduction}

The well-studied Erd\H{o}s-R\'enyi graph, $G(n,p)$, is a basic model for random networks that is amenable to structural analysis. However, $G(n,p)$ is not suited as a model for real-world networks; perhaps the most common criticism is that sparse realizations of $G(n,p)$ do not exhibit clustering~\cite{gl}.  Thus $G(n,p)$ is not a good model for most social networks which are usually sparse and have nontrivial clustering.  In many cases this phenomenon (sparsity together with clustering) is a result of the graph originating as the intersection graph of a larger bipartite graph.  
For example, the well-known collaboration graphs of scientists (or of movie actors) is derived from the bipartite graph of scientists and papers (respectively, actors and movies)~\cite{watts-1998-collective,newman-2001-scientific}. 

A simple natural model for such networks is the random intersection graph.  Random intersection graphs were introduced by Karo\'{n}ski, Scheinerman and Singer-Cohen \cite{singer-1995-thesis,karonski-1999-random} and have recently attracted much attention \cite{bloznelis-2013-assortativity,dk,jks,gjr1,ryb-inhomogenous}. 
We study the phase transition for components in the inhomogeneous random intersection graph model defined by Nikoletseas, Raptopoulos and Spirakis \cite{nrs, nrs1}.  Let $\mathbf{p}=(p_i)_{i=1}^{m}$ be a sequence of $m$ probabilities, $V$ a set of $n$ vertices and $A=\{a_1, a_2, \ldots, a_m\}$ a set of $m$ attributes.  For $(v,a_i)\in V\times A$, define independent indicator random variables $\mc{I}_{v,a_i}\equiv$ Bernoulli$(p_i)$.  A random bipartite graph $B$ is defined on the vertices $V$ and attributes to contain exactly those edges, $(v, a_i)$ for which $\mc{I}_{v,a_i}=1$.  Finally, the random intersection graph $G$ is obtained from the bipartite graph $B$ by projecting onto the vertices $V$: two vertices are connected in $G$ if they share at least one common attribute in $B$.

This paper is concerned with asymptotic results; for each $n$ let $\mathbf{p}^{(n)}$ be a vector of $m=m(n)$ probabilities.  This defines a sequence of random intersection graphs indexed by $n$.  We say an event $E_n$ holds {\it with high probability} if $\mathbb{P}[E_n]\rightarrow 1$ as $n\rightarrow \infty$.  We show that depending on the sequences $\mathbf{p}^{(n)}$ one may observe larger or smaller jumps in the phase transition.  

In previous work \cite{beh, ll}, the phase transition was located for random intersection graphs defined with uniform probabilities.  The component evolution of inhomogeneous random intersection graphs has been studied for a different model of random intersection graphs \cite{blo1, blo2}.  The results in these papers used tools developed by Bollob\'as, Janson and Riordan\cite{bjr} and are exact, though they only consider those cases when the giant component is linear.  In \cite{bollobas_sparse_2011} another general model of sparse random graphs is introduced and analyzed.  
Behrisch~\cite{beh} studied the uniform homogeneous case when all $p_i\equiv p= c/\sqrt{nm}$ and noted that if $p=\omega(1/n)$, then the largest component jumps from size $O(np\log n)$ to $\Theta(p^{-1})$; a smaller jump than observed in Erd\H os R\'enyi random graphs.  On the other hand for $p=O(1/n)$ the largest component jumps from size $O(\log n)$ to $\Theta(n)$; a jump similar to that in Erd\H os-R\'enyi random graphs.  Indeed for $m$ large enough and $p_i\equiv c/\sqrt{mn}$, the random intersection model is equivalent to $G(n,p)$\cite{fss-c, ryb2}.  Our theorems show these phenomena occurring in the more general setting of inhomogeneous random intersection graphs as well.


\begin{theorem}\label{sub-main} If $n\sum p_{i}^{2}<1$ then with high probability all components in $G$ will have size at most $O(\max\{np\log n,\log n\})$ where $p=\max p_{i}$. \end{theorem}

Note that each attribute $a_i$ contributes a clique of expected size $np_i$ to the random intersection graph; thus Theorem \ref{sub-main} is very close to best possible.

\begin{theorem}\label{sup-main}
If $n\sum p_i^2=c>1$ and there exists a $\gamma>1/2$ such that $\max \mathbf{p}^{(n)} = o(n^{-\gamma})$, then with high probability there exists a unique largest component.  This component will have size $(1-\rho)n$ where $\rho$ is the unique solution in $[0,1)$ to the equation 
\begin{equation}\label{GW-multi}
x = \prod_{i=1}^m \big[ 1-p_i(1-(1-p_i(1-x))^n) \big].
\end{equation}
All other components will have size of order $O(\max\{np \log n, \log n\})$.
\end{theorem}

Importantly, the unique largest component guaranteed by Theorem \ref{sup-main} is not necessarily linear in $n$.  Under the conditions of the Theorems \ref{sub-main} and \ref{sup-main}, there is, however, a jump in the size of the largest component when transitioning from the subcritical phase (when $n\sum p_i^2<1$) to the supercritical phase (when $n\sum p_i^2>1$).  Thus a phase transition is observed.  The phase transition is made apparent in comparing Theorems \ref{sub-main} and \ref{sup-main-old} though the later is not necessarily best possible.

As our model is quite general, our theorems do not always give the best possible bounds. What is perhaps surprising is that despite the generality of the model, we can locate the phase transition exactly. No assumptions of uniformity nor convergence of the sequences ${\bf p}^{(n)}$ are necessary; we only require that $n\sum p_{i}^{2}$ be a constant. Because of this generality, there are many cases where the solutions to Equation(\ref{GW-multi}) do not  converge as $n\rightarrow\infty$.  In such cases, even the order of magnitude of the unique largest component may fluctuate. To compensate for this, we state a weaker version of Theorem \ref{sup-main} which gives a lower bound on the order of magnitude of the unique largest component.  We also show, in Proposition \ref{uniform-case}, how to use Theorem \ref{sup-main} to derive the exact size of the largest component in the uniform (homogeneous) case.  In this way we recover exactly previously proved results using our more general method \cite{beh, ll}.

\begin{theorem}\label{sup-main-old}
If $n\sum p_i^2=c>1$ and there exists a $\gamma>1/2$ such that $p = \max p_i = o(n^{-\gamma})$, then with high probability there exists a unique giant component.  If $p^{-1}=o(n)$, this component will be of size at least $\Omega(p^{-1})$.  Otherwise it will be of size $\Omega(n)$.  All other components will have size at most $O(\max\{np \log n, \log n\})$.
\end{theorem}

\begin{example}  Let $m=n^\alpha,\; \alpha<1$, and set $p_i\equiv c/\sqrt{mn}$ for some constant $c$.  If $c<1$ then by Theorem \ref{sub-main}, each component has size at most $O(\sqrt{n/m}\log n)$.  On the other hand, if $c>1$, there exists a unique largest component whose size is $\Omega(np)$ by Theorem \ref{sup-main}.  See Proposition \ref{uniform-case} for the derivation of the exact bound.  These bounds are the same as obtained by Behrisch \cite{beh}.
\end{example}

\begin{example}  Let $m=\beta n$, and set $p_i\equiv c/\sqrt{mn}$ for some constant $c$.  If $c<1$, then each component has size at most $O(\log n)$.  This is the same bound as obtained in \cite{ll}.  On the other hand, if $c>1$ Theorem \ref{sup-main} implies the existence of a unique largest linear component; in Proposition \ref{uniform-case} the exact size is derived.  Here our bounds are the same as previously derived \cite{ll}.
\end{example}

As is standard in the analysis of the phase transition of random graphs we will use both concentration results and the theory of branching processes, specifically Galton-Watson processes.  In the next two sections, we collect the results we will use from these two areas.  
We do not provide proofs for results which are either well known or easily derived from well know results. 

\section{Concentration of Measure}

Recall the Chernoff inequality on $X \sim  \Bin(n,p)$ with $t>0$ (for a proof see \cite{jlr} Theorem 2.1) 
\begin{equation}\label{Chernoff}
\mathbb{P}\left[X\geq \mathbb{E}X+t \right] \leq \exp\left(-\frac{t^2}{2(np+t/3)}\right) \,.
\end{equation}
Given a subset of the attributes $\{a_{i_1},a_{i_2},\ldots, a_{i_k}\}$, it will be useful to approximate the number of vertices likely to be connected to at least one of the given attributes.  To find an upper bound on the number of vertices, we first estimate
$$
W=\sum_{i\in \{i_1, i_2, \ldots, i_k\}} p_i
$$
and then use Equation (\ref{Chernoff}).  To find good approximations of $W$ as above, we need the following useful generalization of Equation~(\ref{Chernoff}) due to McDiarmid~\cite{mcd} and further generalized by Chung and Lu~\cite{cl}.

\begin{theorem}[\cite{cl}]\label{Chung-Lu}
Suppose $Y_i$ are independent random variables with $M_1\leq Y_i\leq M_2$, for $1\leq i\leq n$.  Let $Y=\sum_{i=1}^{n} Y_i$ and $||Y||=\sqrt{\sum_{i=1}^{n} \mathbb{E}(Y_i^2)}$. Then 
\begin{equation}\label{CL-up}
\mathbb{P}[Y\geq \mathbb{E}Y + \lambda]\leq\exp\left(-\frac{\lambda^2}{2(||Y||^2 + M_2\lambda/3)}\right),
\end{equation}

\begin{equation}\label{CL-down}
\mathbb{P}[Y\leq \mathbb{E}Y - \lambda]\leq\exp\left(-\frac{\lambda^2}{2(||Y||^2 - M_1\lambda/3)}\right).
\end{equation}

\end{theorem}

\section{Branching Processes}
We shall make use of the theory of Galton-Watson branching processes.  In a single-type Galton-Watson branching process, each individual has descendants given by a common distribution, $Z$.  Standard results (see \cite{har}, Chapter 1) show that if the mean of $Z$ is less than $1$, the process dies out eventually while if the mean is greater than $1$, there is a positive probability given by $1-\rho$, that the process survives indefinitely.  In this case, $\rho$ is the unique solution in $[0,1)$ to the equation
\begin{equation}\label{GW-single}
x = \sum_{i=0}^\infty \mathbb{P}[Z=i]x^i.
\end{equation}

For a random intersection graph $G$ with parameters $n$ and $\mathbf{p}=(p_i)_{i=1}^m$, we will associate the Galton-Watson process where descendants are taken from the probability distribution of the degree of a random vertex, $v$, in $G$ (i.e. $\mathbb{P}[Z=k] = \mathbb{P}[d(v)=k]$ for each $k$.)   Lemma \ref{probs} elucidates the relationship between $n,\mathbf{p}$ and the associated Galton-Watson process.

\begin{lemma}\label{probs}
Let $n$, $\mathbf{p}=(p_1, p_2, \ldots, p_{m_1})$ and $\mathbf{q}=(q_1, q_2, \ldots, q_{m_2})$ be given such that there exists a $S\subset [m_2]$ and a one-to-one map $\pi:[m_2]\backslash S\rightarrow [m_1]$ with the properties that
\begin{enumerate}
\item[\textrm{(i)}] $\displaystyle\forall j\in [m_2]\backslash S,\; p_{\pi(j)} = q_j$,
\item[\textrm{(ii)}] $\displaystyle\forall j\in S, \forall i\in[m_1]\backslash \pi\big([m_2]\backslash S\big),\; p_i\geq q_j$,
\item[\textrm{(iii)}] $\displaystyle n\sum_{i=1}^{m_1} p_i^2 = n\sum_{j=1}^{m_2} q_j^2$.
\end{enumerate}
If $\mc{X}$ and $\mc{Y}$ are the Galton-Watson processes associated with $\mathbf{p},\mathbf{q}$ respectively, then the probability that $\mc{X}$ dies out is at least as large as the probability that $\mc{Y}$ dies out.
\end{lemma}

\begin{proof}
Let $X$ and $Y$ be the degree distributions for the degrees of an arbitrary vertex in the random intersection graphs with parameters $n,\mathbf{p}$ and $n,\mathbf{q}$ respectively. Then $\mc{X}$ and $\mc{Y}$ are the Galton-Watson processes where each generation is chosen independently from $X$, respectively $Y$.  Writing $x_i=\mathbb{P}[X=i]$ and $y_i=\mathbb{P}[Y=i]$ it is easy to see that 
\begin{equation}\label{mean}
\sum_{i=1}^n i x_i = \sum_{j=1}^n j y_j \,, 
\end{equation}
which is equivalent to condition (iii).  Moreover, conditions (i)-(iii) imply that $x_0\geq y_0$ and that there exists an $l>0$ with
\begin{equation}\label{bounded-pgf}
\forall i, 0<i\leq l,\;\;x_i\leq y_i \text{ and } \forall i, i>l,\;\;x_i\geq y_i.
\end{equation}
It is now easy to show that $f(z)=\sum_{i=0}^\infty x_i z^i$ dominates $g(z)=\sum_{j=0}^\infty y_j z^j$ on the interval $[0,1]$.  
Indeed writing $h(z)=f(z)-g(z)$ we have $h(0) = x_0 - y_0 \geq 0$ while $h(1)=0$. Then Equation (\ref{mean}) and Statement (\ref{bounded-pgf}) imply that for $0<z<1$, $h'(z)\leq h'(1)=0$.  That is, $h$ is decreasing on $[0,1]$, hence $h(z)\geq0$ for $z\in[0,1]$. In particular if the expected values $f'(1)=g'(1)$ are greater than 1, then there is a non-zero probability $1-\rho$ that the process $\mc{Y}$ survives.  In this case, $\rho$ satisfies the equation $\rho = g(\rho)$.  Then $f(\rho)\geq \rho$ which implies that the solution to the equation $z = f(z)$ is at least $\rho$. Then the probability that $\mc{X}$ dies is at least as large as the probability that $\mc{Y}$ dies out.
\end{proof}

\subsection{Multi-type Galton Watson Processes}

It will be convenient to consider multi-type Galton-Watson processes as well.  For a random intersection graph with parameters $n$ and $\mathbf{p}=(p_i)_{i=1}^m$, we associate the following $m+1$ type Galton-Watson process.  
%
Individuals of type $0$ relate to the vertices in the associated random bipartite graph $B$, while all other individuals in this process relate to attributes of $B$.  
Moreover, individuals of type $0$ can have offspring of each of the types $1,2,\ldots, m$; the amount is taken from Bernoulli($p_i$) respectively. Individuals of types $i=1,2, \ldots, m$ only have offspring of type $0$, where the amount is taken from the distribution $\Bin(n, p_i)$. The process starts with one individual of type $0$. 
We review standard results \cite{har} which imply that if $n\sum_i p_i^2=c>1$ then this multi-type process survives with positive probability 1-$\rho$ where $\rho$ is given by the unique solution in $(0,1)$ to Equation~(\ref{GW-multi}).  Note that for a given parameter set, the associated single-type Galton-Watson process and multi-type Galton-Watson process have the same probabilities of survival and extinction.


Consider a general multi-type Galton-Watson process with $m+1$ types labeled $0,1,\ldots,m$. For
each positive integer $N$ and each type $i,$ define $f_{N}^{i}(x_{0},x_{1},\ldots,x_{m})$
to be the generating functions for the descendants at time $N$ given
that the process started with exactly one individual of type $i$.
That is, let $p_{N}^{i}(r_{0},r_{1},\ldots,r_{m})$ represent the
probability that the process starting with one individual of type
$i$ will in the $N^{\text{th}}$ generation have $r_{0},r_{1},\ldots,r_{m}$
offspring of types $0,1,\ldots,m$ respectively. Then the generating
functions can be expressed as
\begin{equation*}
f_{N}^{i}(x_{0},x_{1},\ldots,x_{m})=\overset{\infty}{\underset{r_{0}=0}{\sum}}\;
\overset{\infty}{\underset{r_{1}=0}{\sum}}\cdots
\overset{\infty}{\underset{r_{0}=m}{\sum}}
p_{N}^{i}(r_{0},r_{1},\ldots,r_{m})x_{1}^{r_{1}}x_{2}^{r_{2}}\cdots x_{m}^{r_{m}}.
\end{equation*}
Writing $\mathbf{x}=(x_{0},x_{1},\ldots,x_{m})$ and $\mathbf{f}_{N}(\mathbf{x})=(f_{N}^{0}(\mathbf{x}),f_{N}^{1}(\mathbf{x}),\ldots,f_{N}^{m}(\mathbf{x}))$,
we have
\begin{equation}
f_{N+1}^{i}(\mathbf{x})=f_{N}^{i}(\mathbf{f}_{1}(\mathbf{x})).\label{eq:i}\end{equation}
From the definition, $f_{N}^{i}(\mathbf{0})$ is exactly the probability
of extinction by the $N^{\text{th}}$ generation if the process starts
with one individual of type $i$. Thus $f_{N+1}^{i}(\mathbf{0})\geq f_{N}^{i}(\mathbf{0})$
and in particular, $\lim f_{N}^{i}(\mathbf{0})$ exists and is less
than or equal to 1. Writing $q_{i}=\lim f_{N}^{i}(0)$ we see that
$\mathbf{q}=(q_{0},q_{1},\ldots,q_{m})$ is a solution to the equation\begin{equation}
\mathbf{f}_1(\mathbf{q})=\mathbf{q}.\label{eq:equation}\end{equation}

Let $m_{ij}$ be the expected number of offspring of type $j$ from
an individual of type $i$ and let $\mathbf{M}=(m_{ij})$ be the matrix
of these first moments. Suppose $\mathbf{s}$ is a vector with $|1-s_{i}|\leq1$
for each $i$. Then from Taylor's theorem with remainder we have 
\begin{equation*}
\mathbf{f}_{N}(\mathbf{1}-\mathbf{s})=\mathbf{1}-\mathbf{M}^{N}\mathbf{s}+o(|\mathbf{s}|)\;\;\;\;\;\;|\mathbf{s}|\rightarrow\mathbf{0}
\end{equation*}

Assume that for any such vector $\mathbf{s}$, there exists an $N_{0}$ with $|\mathbf{M}^{N_{0}}\mathbf{s}|>2|\mathbf{s}|$.  Then it follows that there exists a nonnegative solution different from
$\mathbf{1}$ to Equation (\ref{eq:equation}). Indeed fix $\epsilon>0$.
If $\mathbf{q}=\mathbf{1}$, then there exists a sufficiently large
$N$ such that $|\mathbf{1}-\mathbf{f}_{N}(\mathbf{0})|<\epsilon$.
Using $\mathbf{s}=\mathbf{1}-\mathbf{f}_{N}(\mathbf{0})$ we conclude
that \begin{equation*}
|\mathbf{1}-\mathbf{f}_{N+N_{0}}(\mathbf{0})|=|\mathbf{1}-\mathbf{f}_{N_{0}}(\mathbf{f}_{N}(\mathbf{0}))|>|\mathbf{1}-\mathbf{f}_{N}(\mathbf{0})|,\end{equation*}
a contradiction to the fact that $\mathbf{q}-\mathbf{f}_{N}(0)\rightarrow\mathbf{0}$
monotonically as $N\rightarrow\infty$. 

If, in addition to the above assumption, we also assume that $q_{i}<1$
for all $i$, it follows that if $\mathbf{q}_{1}$
is any vector in the unit cube not equal to $\mathbf{1}$, we have
$\mathbf{f}_{N}(\mathbf{q}_{1})\rightarrow\mathbf{q}$ as $N\rightarrow\infty$.
(For a proof, see II.7.2 in \cite{har}.) On the other hand, if there exist
two types, $i$ and $j$ with $q_{i}=1$ and $q_{j}\not=1$, then
Equations (\ref{eq:i}) and (\ref{eq:equation}) imply that for all
all $N$, $1=f_{N}^{i}(\mathbf{q})$. In particular $f_{N}^{i}(\mathbf{x})$
is thus independent of $x_{j}$ which implies that $(\mathbf{M}^{N})_{ij}=0$
for all $N.$

\begin{corollary}\label{GW}
If $n$ and $\mathbf{p}$ are given with $n\sum p_i^2=c>1$ then the associated multi-type Galton-Watson process, as defined above, survives with probability $1-\rho$ where $\rho$ is the unique solution in $[0,1)$ to Equation (\ref{GW-multi}).
\end{corollary}
\begin{proof}
Without loss of generality, we suppose the $p_i$ are all nonzero.  Note that the probability generating functions associated with the multi-type Galton-Watson process are
\begin{equation}
f_1^i(\mathbf{x}) = 
\begin{cases}
(1-p_{i}+x_{0}p_{i})^{n} &\text{if $i>0$}\\
\prod_{i=1}^{m}(1-p_{i}+x_{i}p_{i}) &\text{if $i=0$}.
\end{cases}
\end{equation}
Thus the solution to Equation (\ref{eq:equation}) is exactly given by Equation (\ref{GW-multi}).

To show that this gives the extinction probability of the branching process, it remains to verify the following two assumptions.  First, that for any vector $\mathbf{s}$ sufficiently close to $\mathbf{1}$, there exists an $N$ with $|\mathbf{M}^{N}\mathbf{s}|>2|\mathbf{s}|$.  Secondly,
that for each pair $i,j$ there exists an $N$ such that $(\mathbf{M}^N)_{ij} \not=0$.

Note that for $i,j>0,\; m_{ij}=0$, while for $i>0,\; m_{0i}=p_{i}$,
$m_{i0}=np_{i}$ and for all $i,\; m_{ii}=0$. 
As $n\sum p_{i}^{2}=c>1$ then we have $(\mathbf{M}^{2N})_{00}=c^{N}$
which clearly implies the first statement. Secondly, it is not hard to check
that $(\mathbf{M}^{2k+1})_{ij}\not=0$ when exactly one of $i,j$
are equal to 0. On the other hand, $(\mathbf{M}^{2k})_{ij}\not=0$
when $i,j>0$ and when $i=j=0$. Thus the second assumption above
also holds and we can conclude that the unique solution in $(0,1)$
to Equation (\ref{GW-multi}) is indeed the probability of extinction when the process starts with one
individual of type 0.
\end{proof}


Finally, we show here how to use Equation (\ref{GW-multi}) to derive the size of the largest component in the supercritical phase for the uniform homogeneous case.

\begin{proposition}\label{uniform-case}
Let $m=m(n)$ be a sequence of integers indexed by $n$ and let $c>1$ be given. Define $p=\sqrt{c/mn}$. Then the associated $m+1$ type Galton-Watson process eventually dies out with probability given by

\begin{equation}
\rho=
\begin{cases}
1-(1-\zeta) mp &\text{if $m=o(n)$}\\
\zeta &\text{if $n=o(m)$}\\
\zeta^* &\text{if $m=\Theta(n)$},
\end{cases}
\end{equation}
where $\zeta$ and $\zeta^*$ are the unique solutions in $(0,1)$ to the equations $\exp\left(c(x-1)\right)=x$ and $\exp\left(mp\exp\left(np(x-1)\right)-1\right)=x$, respectively.
\end{proposition}

\begin{proof}
In each case we will use the fact that $1-p=\exp(-p-o(p))$.

When $mp=o(1)$, letting $\rho=1-(1-\zeta) mp$, we have
\begin{align}
\big[ 1-p(1-(1-p(1-\rho))^n) \big]^m &= \left [ 1-p\left(1-e^{-np(1-\rho)(1+o(1))}\right)\right]^m \notag\\
&= \left [ 1-p\left(1-e^{-c(1-\zeta)(1+o(1))}\right)\right]^m\notag\\
&= \left [ 1-p\big(1-\zeta(1-o(1))\big)\right]^m  \label{behrisch-2}\\
&= \exp\left[-mp\big(1-\zeta(1-o(1))\big) (1+o(1))\right] \notag\\
&= 1-\big(1-\zeta(1-o(1))\big) mp(1+o(1)) \rightarrow \rho. \notag
\end{align}

Secondly, if $np=o(1)$ then we have
\begin{align}
\big[ 1-p(1-(1-p(1-\zeta))^n) \big]^m &= \left [ 1-p\left(1-e^{-np(1-\zeta)(1+o(1))}\right)\right]^m \notag\\
&= [1-np^2(1-\zeta)(1+o(1))]^m \label{behrisch-1} \notag\\
&= \exp\big(-c(1-\zeta)(1+o(1))\big) \rightarrow \zeta \,.
\end{align}

Finally, if $n=\Theta(m)$, then $np$ and $mp$ are constants. We have
\begin{align}
&\prod_{i=1}^m \big[ 1-p(1-(1-p(1-\zeta^*))^n) \big] \notag\\
&\leq \left [ 1-p\left(1-e^{-np(1-\zeta^*)(1+o(1))}\right)\right]^m \notag\\
&= \exp\left(mp\left(e^{np(\zeta^*-1)(1+o(1))}-1\right)(1-o(1))\right).\label{lag-lin}
\end{align}
\end{proof}

Note if $n$ is replaced with $n(1-o(1))$ and $m$ with $m(1-o(1))$, the asymptotic results are the same.

\section{Proofs of Main Theorems}
%
%

\subsection{Discovery Process}\label{disc}
For a random intersection graph $G$, we define the following discovery process.  Let $B$ be the bipartite graph associated to $G$ and let $v_1$ be a vertex (as opposed to an attribute) of $B$.   For $i=0,1,2,\ldots$ inductively define sets of unsaturated vertices, discovered vertices and discovered attributes, denoted by $U_i, V_i, A_i$, respectively.  Initially set $A_0=\emptyset,$ and $V_0=U_0=\{v_1\}$.  At step $i$, if $U_{i-1}$ is empty the process terminates.  Otherwise pick $v_i\in U_{i-1}$.  Let $A_i'$ denote the set of attributes connected to $v_i$ in $A\backslash A_{i-1}$.  Thus $A_i'$ is the set of newly discovered attributes.  Discover next the vertices, $V_i'$ of $V\backslash V_{i-1}$ connected to at least one attribute in $A_i'$.  Let $X_i$ denote the cardinality of $V_i'$.  Note again that $X_i$ is the number of newly discovered vertices.  Define the sets 
\begin{align*}
A_i &= A_{i-1}\cup A_i' \\
V_i &= V_{i-1}\cup V_i' \\
U_i &= (U_{i-1}\backslash \{v_i\})\cup V_i' \,.
\end{align*}
A vertex or an attribute can only be discovered once.  Crucially, the event that the vertex $v_i$ is connected to an attribute $a\in A\backslash A_{i-1}$ is independent of the history of the discovery process.  Similarly, the event that an attribute $a\in A_i'$ is connected to a vertex $v\in V\backslash V_{i-1}$ is independent of the history of the discovery process.

%
%
\subsection{Subcritical phase}\label{sub-sec}

\begin{proof}[Proof of Theorem \ref{sub-main}]
Let $\mathbf{p}=(p_i)_{i=1}^{m}$ be given such that $n\sum p_i^2<1$.  Let $G$ be a random intersection graph obtained from $\mathbf{p}$ and set $p=\max p_i$.  Consider the discovery process of $G$: if $A_i'$ is known but $V_i$ has not yet been discovered, define the random variables $W_i = \sum_{j\in A_i'}p_j$ and $X_i^+\sim \Bin(n, W_i)$.  $W_i$ can be thought of as the weight of the attributes associated to the vertex $v_i$. Note that $X_i^+$  stochastically dominate $X_i$, as $$1-\prod_{j\in A_i'} (1-p_j)\leq \sum_{j\in A_i'} p_j.$$

The proof now follows from the following three claims.
\begin{claim}\label{sub-step1}
$\sum_{i=1}^{k}X_{i}$ is stochastically dominated by $X_{(k)}^{+}\sim \Bin(n,\sum_{i=1}^{k}W_{i})$.
\end{claim}
\begin{claim}\label{sub-step2}
Let $k>\frac{10}{(1-c)^2}np\log n$. Then $n\mathbb{P}[\sum_{i=1}^{k}W_i>\frac{k+kc}{2n}]=o(1)$.
\end{claim}
\begin{claim}\label{sub-step3}
Let $k>(15\log n)/(1-c)^2$ and $X_{(k)}^{+}\sim \Bin(n,\sum_{i=1}^{k}W_{i})$.  If $\sum_{i=1}^{k}W_i\leq (k-1+kc)/2n$, then $n\mathbb{P}[X_{(k)}^{+}\geq k-1]=o(1)$.
\end{claim}

Before we prove the claims, we show that they imply the theorem.  First, note that the probability that the component in $G$ containing $v_1$ has size at least $k$ is bounded by $\mathbb{P}[\sum_{i=1}^{k}X_i\geq k-1]$.  Claims \ref{sub-step1} and \ref{sub-step3} imply that if $\sum_{i=1}^{k}W_i$ is small enough then all components have size $O(\log n)$.  However, to prove $\sum_{i=1}^{k}W_i$ is indeed small enough in Claim~\ref{sub-step2} we need $k=\Theta(np \log n)$.  As $k$ is the upper bound on the component sizes, we conclude that all components in $G$ have size at most $O(\max\{np \log n, \log n\})$ as desired.
\end{proof}

We now prove Claims~\ref{sub-step1},~\ref{sub-step2},~\ref{sub-step3}.

\begin{proof}[Proof of \ref{sub-step1}]
It is clear that for each $i$, $X_i$ is stochastically dominated by $X_i^+$. Similarly $\sum_{i=1}^k X_i^+$ is stochastically dominated by $X_{(k)}^+ \sim \Bin(n,\sum_{i=1}^{k}W_{i})$.
\end{proof}

\begin{proof}[Proof of \ref{sub-step2}]

Recall that $W_i$ is the weight of the attributes associated to $v_i$ in the discovery process. As attributes can only be discovered once during the process, $W_i\leq \sum_{j=1}^m p_j \mc{I}_{v_i,a_j}$. In particular $\sum_{i=1}^{k}W_i \leq \sum_{i=1}^{k} \sum_{j=1}^m p_j \mc{I}_{v_i,a_j}$.  
The last sum consists of $k m$ summands each of which is no greater than $p$. Applying Theorem \ref{Chung-Lu} with $M_2=p$ it follows that
\begin{align*}
\left|\left|\sum_{i=1}^{k} \sum_{j=1}^{m} p_j \mc{I}_{v_i,a_j}\right|\right|^2 &= k \mathbb{E}\left[\sum_{j=1}^{m} p_j^2 \mc{I}_{v_i,a_j}\right] = k
\sum_{j=1}^{m} p_{j}^{3} \leq pk\sum_{j=1}^{m} p_j^2 = cpk/n \,.
\end{align*}
Applying Theorem \ref{Chung-Lu} with $\lambda = (1-c)k/(2n)$, we obtain
\begin{align*}
n\mathbb{P}\left[\sum_{i=1}^{k}W_i>\frac{1+c}{2}\frac{k}{n}\right] & \leq n\mathbb{P}\left[\sum_{i=1}^{k} \sum_{j=1}^{m} p_j \mc{I}_{v_i,a_j}>\frac{1+c}{2}\frac{k}{n}\right] \\
& \leq n\exp \left(-\frac{ (1-c)^2 k^2 }{(2n)^2 2\left(cpk/n + p(1-c)k/(6n)\right)} \right)\\
& = n\exp\left(-\frac{(1-c)^2k}{2np(4c+1)}\right) = o(1) \,,
\end{align*}
where the last equality follows for $k>\frac{10}{(1-c)^2}np\log n$.
\end{proof}

\begin{proof}[Proof of \ref{sub-step3}]
As $X_{(k)}^+\sim \Bin(n,\sum_{i=1}^{k}W_i)$ and $\sum_{i=1}^{k}W_i\leq (k+kc)/2n$, it follows that $X_{(k)}^+$ is stochastically dominated by $X_{(k)}^{++}\sim \Bin(n,(k+kc)/2n)$.  By Chernoff's inequality, 

\begin{align*}
\mathbb{P}[ X_{(k)}^{++} \geq k-1] &\leq \mathbb{P} \left[ X_{(k)}^{++} \geq \frac{(1+c)k}{2} + \frac{(1-c)k}{2} -1 \right] \\
&\leq \exp \left( -\frac{\left(\frac{1-c}{2}k-1\right)^2}{(1+c)k+2\left(\frac{1-c}{2}k-1\right)/3}\right) \\
&\leq \exp \left( - \frac{(1-c)^2k}{5(2+c)} \right) = o(1) \,,
\end{align*}
%
where the last equality follows by letting $k>(15\log n)/(1-c)^2$.
\end{proof}

%
%

\subsection{Supercritical phase}\label{sup-sec}
\begin{proof}[Proof of Theorem \ref{sup-main}]
Let $\gamma \in (1/2,2/3)$ and $\mathbf{p}=(p_i)_{i=1}^{m}$ be given such that $p=\max p_i = o(n^{-\gamma})$ and $\sum p_i^2 = c/n$ with constant $c>1$.
Let $G$ be the random intersection graph obtained. Consider the same discovery process as defined in Section \ref{disc} on the associated random bipartite graph $B$. In particular, recall that $A_i'$ is the set of newly discovered attributes at step $i$ and that $W_i = \sum_{j\in A_i'} p_j$. Let $k_-=\max\{ \frac{5npc}{(1-c)^2} \log n, \frac{125c}{(1-c)^2}\log n\}$ and $k_+=n^{\gamma}$. The following is an adaptation of standard results for the phase transition in Erd\H os-R\'enyi random graphs (Theorem 5.4, \cite{jlr}).

First note that for each $k \in [k_{-}, k_{+}]$, the following holds
\begin{equation}\label{sup-1}
\mathbb{E}\left[\sum_{i=1}^{k}W_i\right] \geq \frac{kc}{n}(1-o(1)).
\end{equation}
To see this, note that the probability that attribute $j$ is discovered by the $k$th 
step is $1-(1-p_j)^k$.  Thus 
\begin{align*}
\mathbb{E}\left[\sum_{i=1}^{k}W_i\right] & =\sum_{j=1}^{m} p_j(1-(1-p_j)^k) \\
&= \sum_j \left(kp_j^2 - \binom{k}{2}p_j^3 + \cdots + \binom{k}{k}(-p_j)^{k+1} \right) \\
& \geq k\sum_j \left(p_j^2- \frac{k-1}{2} p_j^3\right)\\
& = k\frac{c}{n}(1-o(1)) \,.
\end{align*}
The last equality follows from $p_j=o(n^{-\gamma})$ which implies $pk=o(1)$.

We now use Equation (\ref{sup-1}) to show that for each $k\in [k_{-}, k_{+}]$, 
\begin{equation}\label{sup-2}
\mathbb{P}\left[\sum_{i=1}^{k}W_i\leq \frac{ck}{n} - \frac{(c-1)k}{3n}\right] = o(n^{-5/3}) \,.
\end{equation}
This follows from (\ref{CL-down}) by writing $\sum_{i=1}^{k}W_i=\sum_{j=1}^m p_j I_j$, with $I_j$ the indicator random variable equal to $1$ with probability $1-(1-p_j)^k$ and $0$ otherwise. Clearly, $p_j I_j\geq 0$ for each $j$ and $||\sum_{i=1}^{k}W_i||^2 = \sum_j \mathbb{E}[(p_j I_j)^2] = \sum_j p_j^2 [1-(1-p_j)^k]$.  Thus
\begin{align*}
\mathbb{P}\left[\sum_{i=1}^{k}W_i\leq \frac{ck}{n} - \frac{(c-1)k}{3n}\right] & \leq \exp\left[-\frac{(\frac{c-1}{3})^2 k^2/n^2}{2\sum_j p_j^2 [1-(1-p_j)^k] }\right]\\
& \leq \exp\left(-\frac{(c-1)^2k^2}{18n^2 \sum_j k p_j^3 }\right)\\
& \leq    \exp\left(-\frac{(c-1)^2 k}{18cpn}\right)=o(n^{-5/3}).
\end{align*}

Indeed, it follows from (\ref{CL-up}) and a similar derivation that for $k \in [k_{-}, k_{+}]$
\begin{equation}\label{sup2-2}
\mathbb{P}\left[\sum_{i=1}^{k}W_i> (2c-1)\frac{k}{n} \right] = o(n^{-5/3}).
\end{equation}

We now show that with high probability there are no components with $k \in [k_{-}, k_{+}]$ vertices. In particular, we show that either the discovery process terminates after $k_-$ steps, or that for each $k\in [k_{-}, k_{+}]$, there are at least $(c-1)k/2$ unsaturated vertices.  As $\sum_{i=1}^{k} X_i = |\mc{U}_k|+k-1$, it will be enough to show that for each $k\in [k_{-}, k_{+}]$, with high probability, $\sum_{i=1}^{k} X_i$ is at least $(c+1)k_+/2.$  From Equation (\ref{sup-2}), 
with high probability the weight of the discovered attributes after $k$ steps will be at least $(2ck+k)/3n$.    

As we only need to find $(c+1)k_+/2$ vertices, we can bound each $X_i$ from below by $X_i^-\sim \Bin(n-(c+1)k_+/2, W_i)$.  We further bound from below $\sum_{i=1}^k X_i^-$ by $X^{-}_{(k)}$ where $X^{-}_{(k)} \sim \Bin(n-(c+1)/2, p_{k}^-)$, where $p_{k}^-$ is defined as
$$
p_{k}^- = \sum_{i=1}^{k}W_i - \frac{1}{2}\left(\sum_{i=1}^{k}W_i\right)^2.
$$
Equation (\ref{sup2-2}) implies that with high probability $p_k^- = \sum_{i=1}^k W_i(1-o(1))$.  In turn, Equation (\ref{sup-2}) then implies that 
$$
p_k^- \geq \frac{(2c+1)k}{3n}(1-o(1))\,.
$$
The probability there is a component of size between $k_-$ and $k_+$ is thus bounded above in the following manner:

\begin{align*}
n\sum_{k=k_-}^{k_+} \mathbb{P}\left[\sum_{i=1}^k X_i\leq k-1 +\frac{(c-1)k}{2} \right] & \leq n\sum_{k=k_-}^{k_+} \mathbb{P}\left[X_{(k)}^- \leq k-1 +\frac{(c-1)k}{2}\right]\\
& \leq n\sum_{k=k_-}^{k_+} \exp\left(-\frac{(c-1)^2k^2}{25(2c+1)k}\right)\\
&\leq nk_+ \exp\left(-\frac{(c-1)^2k_-}{25(2c+1)}\right)  =o(1).
\end{align*}

We now show that if there is a component of size at least $k_+$, then it is unique with high probability.  Suppose that there are two vertices $v'$ and $v''$ which belong to components of size at least $k_+$.  Consider our discovery process starting at $v'$.  At the end of the $k_+$ step, there are at least $(c-1)k_+/2$ unsaturated vertices which belong to the component containing $v'$.  Similarly, if we consider the discovery process starting with $v''$, again there are at least $(c-1)k_+/2$ unsaturated vertices in the component containing $v''$.  Denote by $A'$ and $A''$ the sets of attributes discovered by the $k_+$ step for each of the two discovery processes, respectively. If the two components are distinct, then in particular none of the unsaturated vertices in $V'$ are connected to any of the unsaturated vertices in $V''$.  With high probability this will not occur.  Indeed,  
\begin{align*}
\mathbb{P}[V'\not\sim V''] & \leq \prod_{i\not\in (A'\cup A'')} (1-p_i)^{(c-1)k_+} \leq \left[e^{-\sum_{i\not\in (A'\cup A'')} p_i}\right]^{k_+(c-1)}\\
& \leq \exp\left(-\frac{(c-1)k_+}{\sqrt{n}}\right) = \exp\left(-(c-1)n^{\gamma-\frac{1}{2}}\right) = o\left(\frac{1}{n^2}\right).
\end{align*}

The last inequality follows from the fact that $\sum_{i\not\in (A'\cup A'')} p_i^2\geq\frac{1}{n}$ implies $\sum_{i\not\in (A'\cup A'')} p_i\geq\frac{1}{\sqrt{n}}$.  Note that Equation (\ref{sup2-2}) implies that 
$n\sum_{i\not\in (A'\cup A'')} p_i^2=c-o(1)>1$.

We have yet to show that $G$ contains a component of size at least $k_+$.  Denote by $\rho=\rho(n,\mathbf{p})$ the probability that a given vertex of the random intersection graph will be in a small (i.e. of size at most $k_-$) component. Now $\rho$ is bounded from below by the extinction probability $\rho_-=\rho_-(n,\mathbf{p})$ of the associated multi-type branching process.  

To bound $\rho$ from above, recall Equation (\ref{sup2-2}) which implies that after $k_-$ steps of the discovery process, with high probability $W_{k_-}\leq(2c-1)k_-/n$.  Thus we bound $\rho$ by  $\rho_+ =\rho_+(n-k_-, \mathbf{p}')+o(1)$ where $\mathbf{p}' = \{p_i | i \in A\backslash A'\}$ for a suitable set of attributes $A'$ such that $\sum_{i\in A'}p_i\leq (2c-1)k_-/n$.  Lemma \ref{probs} implies that $\rho_+$ is largest when $A'$ consists of the smallest (by weight) elements of $A$. Thus assuming without loss of generality that the sequence $\mathbf{p}$ is monotone increasing, let $l$ be maximal such that $\sum_{i=1}^{l}p_i<(2c-1)k_-/n$.  Then if $A'$ consists of the first $l$ attributes, and $\mathbf{p}'=(p_i)_{i>k}^m$, $\rho_+$ will be largest given that $\sum_{i\in A'}p_i\leq (2c-1)k_-/n$.

The definition of $l$ implies that $n\sum_{i=1}^l p_i^2 = o(1)$ and thus $\sum_{i=1}^l \max\{p_i, np_i^2\} = o(1)$.  It follows from Equation (\ref{GW-multi}) that in the limit as $n\rightarrow \infty$ $\rho_+= \rho_-(1+o(1))$.  Thus $Y$, the expected number of vertices in small components of $G$ is $\rho(1 +o(1))n$.  To see that $Y$ is strongly concentrated about its mean, note that 
\begin{equation}
\nonumber
\mathbb{E}[Y^2] \leq n\rho(n,\mathbf{p})k_- + n\rho(n, \mathbf{p})n\rho(n-k_-, \mathbf{p})=(1+o(1))\mathbb{E}[Y]\,.
\end{equation}
Thus by Chebyshev's inequality the variance of $Y$ is $o(1)$ as desired.
\end{proof}

\begin{proof}[Proof of Theorem \ref{sup-main-old}]
The proof is exactly the same as above except that we give a weaker upper bound for $\rho$.  Let $G_1$ be the random intersection graph with parameters $n-k_-, \mathbf{p}'=(p_i)_{i>k}$ as above.  Let $Y$ be the degree distribution of $G_1$ and $\mathcal{Y}$ the associated single-type Galton-Watson branching process.  Let $G_2$ be the random intersection graph on $n-k_-$ vertices and $m=\lfloor \frac{c}{np^2}\rfloor$ attributes, each assigned the probability $p$.  Let $X$ be the degree distribution of $G_2$ and $\mathcal{X}$ the corresponding Galton-Watson process.  The probability generating function for $\mathcal{X}$ and $\mathcal{Y}$ will satisfy the conditions of Lemma \ref{probs} and thus the extinction probability for $\mathcal{X}$ gives an upper bound on the extinction probability for $\mathcal{Y}$ and thus an upper bound on the probability a vertex in $G_1$ is in a small component.  Applying Proposition \ref{uniform-case} it follows that the expected size of the largest component in $G$ is at least $\Omega(\min\{p^{-1}, n\})$. 
\end{proof}


\bibliographystyle{amsplain}
\bibliography{intersection-graphs}

\end{document}